\documentclass[12pt]{article}

\usepackage[english]{babel}
\usepackage{amsmath,amsthm}
\usepackage{amsfonts}
\usepackage{graphicx}
\usepackage{epsfig}
\usepackage{epstopdf}
\usepackage{url}
\usepackage{amssymb}
\usepackage{color}
\usepackage[%
breaklinks%
,colorlinks%
,linkcolor=red
,anchorcolor=blue%
,pagecolor=blue%
,citecolor=blue%
,bookmarks=false%
]{hyperref}

\newtheorem{thm}{Theorem}[section]

\newtheorem{lem}[thm]{Lemma}
\newtheorem{prop}[thm]{Proposition}
\theoremstyle{Definition}

\newtheorem{rem}{Remark}
\newtheorem{alg}{Algorithm}
\newtheorem{exam}{Example}
\numberwithin{equation}{section}

\def\P{{\mathbf{P}}}
\def\L{{\mathcal{L}}}

\def\X{{\mathbf{X}}}

\def\C{{\mathcal{C}}}
\def\S{{\mathcal{S}}}
\def\N{{\mathbf{N}}}
\def\num{{\mathrm{num}}}
\def\Res{{\mathrm{Res}}}
\def\para{\vspace{4 mm}}

\begin{document}

\title{Determination and (re)parametrization of rational developable surfaces}%

\author{Sonia P\'erez-D\'{\i}az  \\
Dpto de F\'{\i}sica y Matem\'aticas \\
        Universidad de Alcal\'a \\
      E-28871 Madrid, Spain  \\
sonia.perez@uah.es
\and Li-Yong Shen\\
School of Mathematical Sciences  \\
      University of CAS \\
        Beijing, China \\lyshen@ucas.ac.cn
      }
\date{}          
\maketitle      

\begin{abstract}The developable surface is an important surface in computer aided design,  geometric modeling and industrial manufactory. It is often given in the standard parametric form, but it can also be in the implicit form which is commonly used in algebraic geometry. Not all algebraic developable surfaces have rational parametrizations. In this paper, we
focus on the rational developable surfaces. For a given algebraic surface, we first determine
whether it is developable by geometric inspection, and we give a rational proper parametrization for the affirmative case.  For a rational parametric surface, we can also determine the developability and give a proper reparametrization for the developable surface.

\end{abstract}

{\bf Keywords:} rational developable surface, parametrization, reparametrization

\section{Introduction}
A developable surface can be constructed by bending a planar region at every point.
It is a commonly used surface in computer aided design and geometric modeling~\cite{ liu06,li13,potman99,sun96}.
Developable surfaces have zero Gaussian curvature and they are a subset of ruled surfaces. In general design, the developable surface is often proposed as a parametric form. In recent years, people challenge to geometrically design with algebraic surfaces since they have more geometric features and topologies than those of the parametric surfaces (see \cite{handbook, turk02}).
In this situation, a natural problem is to determine the type of algebraic modeling surfaces. For expected cases, some surfaces can be commonly used surfaces, for example, developable surfaces. As a successive problem, we need to find a rational parametrization of the determined surface, since the parametrizations are better representations for manufactural control and computer display.

The two problems are both difficult for general surfaces, particularly in computation. Since the algebraic surfaces are basic objects in algebraic geometry, there were some classical results associated to these two problems. Let $\cal S$ be an algebraic surface. If $\cal S$  is a rational surface then $P_n=q=0$ for all $n$, and conversely, any surface with
$q = P_2 = 0$ is a rational surface, where $P_n$ and $q$ are the \emph{plurigenus} and the \emph{irregularity} of $\cal S$, respectively. This is called Castelnuovo's rationality criterion (see~\cite{beauville96}, V.1).
If $\cal S$  is a ruled surface then $P_n=0$ for all $n$,  and conversely, any surface with $P_4 = P_6 = 0$ (or $P_{12}=0$) is a ruled surface. This is called the criterion of ruled surfaces (see~\cite{beauville96}, VI.18).

\para

The above results gave important theoretical effects, but there were lack of practical method for real computation. In fact, the plurigenus and irregularity are difficult to compute.
Therefore, for a general implicit curve and surface, to propose a parametrization algorithm is still an open problem~\cite{handbook}.
To meet the practical needs, people tried to design some parametrization algorithm for some special surfaces which are commonly used.
Sederberg and Snively~\cite{seds87} proposed four methods of parametrization of
cubic algebraic surfaces. Sederberg~\cite{sed90b} and Bajaj et al.~\cite{bajaj98} expanded this method.
In~\cite{handbook02}, a method to parameterize a quadric is given using
a stereographic projection; Berry et al.~\cite{berry01} tried to unify the implicitization and parametrization of a nonsingular cubic surface with Hilbert-Burch matrices.
Recently, Chen et al.~\cite{chen06} presented a method to deal with the implicitization
and parametrization of quadratic and cubic surfaces by the $\mu$-basis which is a developing method.  In~\cite{sonia13}, P\'erez-D\'iaz and Shen characterized the rational ruled surfaces using the reduced standard form. 
These methods were designed for some special surfaces. For a general given surface,   Schicho~\cite{schicho97} gave well analysis in parametrization problem. He provided more contributions on theoretical analysis than practical computation, since the problem is quite difficult for general situations.

\para

There were also some numerical mesh parametrization methods designed for the industrial manufactory~\cite{mf05}. One  difficult problem in the numerical methods is to set the values of parameter for the points in an implicit surface. This is the main reason that people can only get an approximate parametrization using the numerical methods.
Since the numerical approximate method may lose some intrinsic properties of the surfaces, we prefer to find the symbolic parametrization method in this paper for a typical surface named as the developable surface.

\para

As mentioned above, the developable surface is an important modeling surface. But there had few papers that discussed the parametrization of an algebraic developable surface. This motives us to focus on the problem in this paper. In the geometric investigation,
a developable surface must be either a cylindrical surface, a conical surface, or a tangential surface of a space curve. We then reduce the problem to determine and parameterize these three special surfaces.

\para

The paper is organized as follows. Some necessary notations and preliminary results are proposed in Section 2. The algebraic rational developable surfaces are characterized in Section~3, and a rational proper parametrization is computed. Examples are given for some typical surfaces. In Section~4, we focus on the parametric surface and the reparametrization problem. The examples are also presented. Finally, we concluded the paper in Section~5.

\section{Preliminaries}

Let $\mathbb{L}[t]$ be the polynomial ring over the subfield $\mathbb{L}$ of an algebrically closed field of characteristic zero $\mathbb{K}$, and $\mathbb{L}(t)$ be the field of rational functions over~$\mathbb{L}$.

\para

A \emph{ruled surface} is defined by  one parameter family of straight lines moving along a curve. The curve is called the \emph{directrix} and the straight lines are called \emph{rulings} or \emph{generators}. A \emph{developable surface} is a ruled surface with zero Gaussian curvature (see~\cite{Carmo, Spivak M79}). If the rulings all pass through one point which called \emph{apex}, the surface is a \emph{conical surface}.  If the rulings of a developable surface are parallel to the same straight line, the surface is a \emph{cylindrical surface}. In the remaining cases the developable surface is the \emph{tangential surface} which is defined by the tangents to a certain space curve.  We also call it the tangential surface of the space curve. The space curve is called \emph{cuspidal edge} of the tangential surface.

Although the developable surfaces are often given in parametric form, not all developable surfaces have rational parametrizations.

 A proper parametrization of a rational ruled surface  in standard form is given by
\begin{equation}\label{Eq1}
\P(s,t) = \P_0(t)+ s \P_1(t).
\end{equation}
where  $\P_i(t)\in \mathbb{R}(t)^3,\,i=1,2,$  and $\P_1\not=(0,0,0)$.
The rational developable surface has three forms by the following
lemma presented in~\cite{Spivak M79}.

\para

\begin{lem} \label{lem_develop} A ruled surface  of the form~\eqref{Eq1} is a developable
surface if and only if $\P_0\times\P_1'\cdot\P_1=0$. In addition, a rational developable surface can only be one of the following cases:
\begin{itemize}
  \item If $\P_0(t)$ is a
  constant vector, then $\P(s,t)$ defines a conical surface.
  \item If $\P_1(t)$ is  a
  constant vector, then $\P(s,t)$ defines a cylindrical surface.
    \item If $\P_1(t)=\P_0'(t)$, then $\P(s,t)$ defines a tangential surface.
\end{itemize}
\end{lem}

\para

For a tangential surface $\S$ defined by the parametrization $\P(s,t) = \P_0(t)+ s \P_0'(t)$, the cuspidal edge $\P_0(t)$ defines a singular curve of $\S$. Since one can find that $\P(s,t)$ is singular at $(0,t)$.

\para
The parametric form is widely used in computer aided geometric design and in geometric modeling. An algebraic surface defined by $F(x,y,z)=0$ may not have a rational parametrization. If an algebraic developable surface has a rational parametrization, we call it a  rational developable surface. In the following of this paper, we focus on finding the rational parametrization of a given rational developable surface.

\section{Implicitly rational developable surface}
We start with a theorem that determines  the developability of  an algebraic surface. The referenced discussion can be found in~\cite{burr50,Carmo, Spivak M79}.
\para

\begin{thm}\label{developability}
  Let $\S$ be an algebraic surface  defined implicitly by the polynomial $F(x,y,z).$
 $\S$ is a developable surface if and only if $K(x,y,z)=0$ on $\S$, where
\begin{equation}\label{kxyz}
  K(x,y,z)=\left|\begin{array}{cccc}
F_{xx}&F_{xy}&F_{xz}&F_{x}\\
  F_{yx}&F_{yy}&F_{yz}&F_{y}\\
  F_{zx}&F_{zy}&F_{zz}&F_{z}\\
  F_{x}&F_{y}&F_{z}&0
\end{array}\right|.
\end{equation}
\end{thm}
  In some papers, it is also said that an algebraic surface is  developable if and only if its Gaussian curvature $\kappa(F)=0$  on $\S$, since we have the formula $\kappa(F)=K(x,y,z)/|\nabla F|^4$, where $\nabla$ means the gradient~\cite{Spivak M79}.  Goldman~\cite{Goldman05} gave a proof for this formula.\\

There are three types of developable surfaces.  We now discuss the rationality for each of them.

\para

\begin{lem}\label{cone}
Let $\S$ be a conical surface with  the apex $\P_0$. Let $\L$ be a plane not passing through $\P_0$, and let $\C$ be the intersection curve of $\S$ and $\L$.
  $\S$  has a rational proper parametrization   of the form  $ \P_0+ s \P_1(t)$ if and only if $\C$  is rational.
\end{lem}
\begin{proof}
For the necessity, let $\S$ be a conical developable surface. $\S$  has a rational
proper  parametrization  $$ \P(s,t)=(p_{01},p_{02},p_{03})+s(p_{11}(t),p_{12}(t),p_{13}(t))\in \mathbb{R}(s,t)^3. $$ Let $\L$ be a plane not passing through the apex $\P_0$ of $\S$, and  we assume  its implicit equation is given as $L(x,y,z)=0$. Substituting $\P(s,t)$ into $L(x,y,z)=0$, one can solve $s=q(t)\in \mathbb{R}(t)$ because $(p_{11}(t),p_{12}(t),p_{13}(t))\neq (0,0,0)$ and $s$ is linear in the equation.
Then, the curve $\C$, given by the intersection of $\L$ and  $\S$,  has a proper rational parametrization defined as $\tilde \P(t)=\P(q(t),t)$, since $\P(s,t)$ is proper and $\L$ does not pass through the apex.

For the sufficiency,  according to the arguments, the apex $\P_0$ is not on $\C$. Suppose that $\C$ has a rational proper parametrization $\tilde \P(t)$. Thus,   $(1-s)\P_0+s\tilde \P(t)$ is a rational proper parametrization of  $\S$, since it defines a conical surface covered $\S$.
\end{proof}

\para

\noindent
For the cylindrical surface, we have the similar property.
\para
\begin{lem}\label{cylind}
Let $\S$ be a cylindrical surface with  the ruling direction $\P_1$. Let  $\L$ be a plane not parallel to  $\P_1$, and let $\C$ be the intersection curve of $\S$ and $\L$.
 $\S$  has a rational proper parametrization of the form  $ \P_0(t)+ s \P_1$ if and only if $\C$ is rational.
\end{lem}
\begin{proof}
For the necessity, let $\S$ be a conical developable surface. $\S$  has a rational
proper  parametrization  $$ \P(s,t)=(p_{01}(t),p_{02}(t),p_{03}(t))+s(p_{11},p_{12},p_{13})\in \mathbb{R}(s,t)^3. $$Let $\L$ be a plane not parallel to $\P_1$, and we assume  its implicit equation is given by $L(x,y,z)=0$. Substituting $\P(s,t)$ into $L(x,y,z)=0$, one can solve $s=q(t)\in \mathbb{R}(t)$  because $(p_{11}(t),p_{12}(t),p_{13}(t))\neq (0,0,0)$ and $s$ is linear in the equation.
Then, the intersection curve $\C$  has a proper rational parametrization, $\tilde \P(t)=\P(q(t),t)$, since $\P(s,t)$ is proper and $\L$ is not parallel to the rulings.

 For the sufficiency, we know that $\C$ is not  a ruling since $\L$ is not parallel to the ruling direction. If $\C$  has a rational proper parametrization $\tilde \P(t)$, then $\P(s,t)=\tilde\P(t)+s \P_1$ is a rational proper parametrization of  $\S$, since $\P(s,t)$ defines a  cylindrical  surface covered~$\S$.
\end{proof}

A tangential developable surface is generated by the tangent lines of a space curve. The intersection of a tangent developable with the normal plane at a point $P$ of the curve generally has a cusp at that point. Thus the tangential developable surface of a space curve has a cuspidal edge along the curve (see~\cite{cleave80}), and the cuspidal edge is a singular curve of the tangential developable surface. In this paper, the space curve is assumed not to be a planar curve in ${\Bbb L}^3$. Since the tangential surface with a planar cuspidal edge is just a plane.
\para

\begin{lem}\label{tangent}
Let $\S$ be a tangential surface.
$\S$  has a rational proper parametrization if and only if  it has a singular curve having a rational proper parametrization $\P_0(t)$, and $ \P_0(t)+ s \P_0'(t)$ is a proper parametrization of $\S$.
\end{lem}
\begin{proof}
For the necessity,   let $\S$ be a tangential surface with a rational
proper  parametrization  $$\P(s,t)=(p_{01}(t),p_{02}(t),p_{03}(t))+s(p'_{01}(t),p'_{02}(t),p_{03}'(t))\in \mathbb{R}(s,t)^3. $$
Then  the rational curve defined by the parametrization $\P_0(s)$ is the cuspidal edge and it is a singular curve of  $\P(s,t)$ since $(s,0)$ is always singular.

The sufficiency is obtained by the construction of the tangent developable surface.
\end{proof}

\para

 We observe that if $F(x,y,z)=0$ defines a tangent developable surface, then the cuspidal edge is a singular curve. Therefore, it is included in the singular set  defined by the algebraic system $S=\{F=0, F_x=0, F_y=0, F_z=0\}$.

\para

In the following, we summarize Theorem~\ref{developability}, and Lemmas~\ref{cone}, \ref{cylind}, \ref{tangent}, and we get the following theorem.

\para

\begin{thm}\label{implicitdeve}Let $\S$  be an algebraic surface defined implicitly by the polynomial $F(x,y,z)$. $\S$  is a rational developable surface if and only if  the following statements hold:
    \begin{itemize}
       \item[1.] $K(x,y,z)=0$, for all points $(x,y,z)$ of $\S$.
       \item[2.] One of the following statements holds:
  \begin{itemize}
    \item[2.1.] $\S$ is a conical surface with apex $\P_0$, and there exists a planar curve $\C \subset \S$ not passing through $\P_0$ and having a proper rational parametrization $\tilde \P(t)$. Furthermore, $(1-s)\P_0+s\tilde \P(t)$ is a proper parametrization of $\S$.
  \item[2.2.] $\S$ is a cylindrical surface with ruling direction $\P_1$, and there exists a planar curve $\C \subset \S$ not parallel to  $\P_1$ and having a proper rational parametrization $\tilde \P(t)$. Furthermore, $\tilde \P(t)+s\P_1(t)$ is a proper parametrization of $\S$.
 \item[2.3.] $\S$ is a tangential surface, and there exists  a space singular curve $\C \subset \S$ having a rational proper parametrization $\P_0(t)$. Furthermore, $ \P_0(t)+ s \P_0'(t)$ is a proper parametrization of $\S$.
      \end{itemize}
  \end{itemize}
\end{thm}

\para
\begin{center}
\sf  Parameterize the developable surfaces
\end{center}
\para

By Theorem~\ref{implicitdeve}, before parametrizing a rational developable surface,  we need to determine the types of the surface: conical, cylindrical or tangential surface.
The normal vector of  $\S$ at $(x,y,z)$ is $\N(x,y,z)=(F_x, F_y, F_z)$, where $F_{var}$  is the partial derivative of $F$ with respect to the variable $var$.  Then, the tangent surface of $\S$ at the point $(u,v,w)$ is given by the equation $T(x,y,z)=F_{u}(x-u)+F_{v}(y-v)+F_{w}(z-w)=0$.

\para

\begin{prop}Let $(F_x, F_y, F_z)$ and $T(x,y,z)=0$ be the normal vector and tangent surface of the developable surface $F(x,y,z)=0$, respectively. It holds that:
\begin{itemize}
 \item[1.] If $T(x,y,z)=0$ passes through a fixed point $\P_0$, $F(x,y,z)=0$ defines a conical surface with the apex $\P_0$.
 \item[2.] If there exists  $(0,0,0)\neq\P_1=(p_{11},p_{12},p_{13})\in \mathbb{R}^3 $ such that $p_{11}F_x+ p_{12}F_y+p_{13}F_z=0$, $F(x,y,z)=0$ defines a cylindrical surface with the ruling direction $\P_1$.
\item[3.] Otherwise, $F(x,y,z)=0$ defines a tangential surface, and its cuspidal edge is included in $\{F=0, F_x=0, F_y=0, F_z=0\}$.
\end{itemize}
\end{prop}
\begin{proof} The equation $F(x,y,z)=0$ defines a developable surface, and there are three different types of these surfaces. According to the definitions, the  conical surface is the only that has an apex such that any tangent surface $T(x,y,z)=0$ passes through it. The situation of the normal direction orthogonal with a constant vector can only happen with  the cylindrical surface. The remain developable surfaces are the tangential surfaces.
\end{proof}

\begin{alg}\label{alg1} {\sc Input}: An algebraic surface $\S$ defined implicitly by  $F(x,y,z)$.\\
  {\sc Output}: A proper rational parametrization $\P(s,t)$ of the rational developable surface $\S$ or a message for $\S$.
 \begin{itemize}
   \item[1.] Compute $K(x,y,z)$ of the form~\eqref{kxyz}. If it is zero on $\S$ go to Step 2. Otherwise, {\sf Return} ``$\S$ is not a developable surface."
   \item[2.] If the tangent surface defined by the equation $T(x,y,z)=0$ passes through a fixed point $\P_0$,  let $\L$ be a plane not passing through $\P_0$, and let $\C$ be the intersection curve of $\S$ and $\L$.
       \begin{itemize}
          \item[2.1.] If $\C$  has a rational proper parametrization $\tilde\P(t)$,  {\sf Return} $ (1-s)\P_0+ s \tilde \P(t)$ is a rational proper parametrization  of  the conical surface $\S$.
         \item[2.2.] Otherwise, {\sf Return} ``$\S$ is a conical surface but not rational."
       \end{itemize}

       \item[3.] If there exists  $(0,0,0)\neq\P_1=(p_{11},p_{12},p_{13})\in \mathbb{R}^3 $ such that $p_{11}F_x+ p_{12}F_y+p_{13}F_z=0$,  let $\L$ be a plane not parallel to $\P_1$, and let $\C$ be the intersection curve of $\S$ and~$\L$.
       \begin{itemize}
          \item[3.1.] If $\C$  has a rational proper parametrization $\tilde\P(t)$,   {\sf Return} $\tilde \P(t) + s\P_1$ is a rational proper parametrization  of  the cylindrical surface $\S$.
         \item[3.2.] Otherwise, {\sf Return} ``$\S$ is a cylindrical developable surface but not rational."
       \end{itemize}

\item[4.] Solve the algebraic system $S=\{F=F_x=F_y=F_z=0\}$ by applying for instance Wu's zero decomposition (see~\cite{Wu00}). Compute a rational proper parametrization $\tilde\P_i(t)$ of a curve $\C_i\in S$ applying for instance the resolvent method in~\cite{Chou1992}.
       \begin{itemize}
          \item[4.1.] If $\tilde\P(t)+s\tilde\P'(t)$ is a rational proper parametrization of $\S$,  {\sf Return} $\tilde\P(t)+s\tilde\P'(t)$ is a rational proper parametrization  of  the tangential surface~$\S$.
         \item[4.2.] If there not exists  any curve $\C\in S$ satisfying the condition of Step~4.1., {\sf Return} ``$\S$ is not a rational developable surface."
       \end{itemize}
 \end{itemize}
\end{alg}

\para

\begin{rem}\label{rem1}We here give some details for computation.
\begin{itemize}
  \item[a.] To find the fixed point $\P_0$ in Step~2, we need to solve  $S=\{T(x,y,z; u,v,w)=0, F(u,v,w)=0\}$ with respect to $\{x,y,z\}$. In order to simplify the computation, we can solve  the linear system $T(x,y,z; u_i,v_i,w_i)=0$ for some random selected points $(u_i,v_i,w_i)\in \S$, and check the solutions lying in  $S$. Observe that one also may use the arithmetic in the quotient field of rational functions ${\Bbb C}(\S)$, and to compute remainders with the polynomial $F(u,v,w)$.
      \item[b.]
 To find the ruling direction $\P_1$ in Step~3, we only need to consider the coefficient vectors of $F_x,F_y$ and $F_z$. It holds that for   a cylindrical surface, there are linearly dependent with associated vector $(p_{11},p_{12},p_{13})$.
\end{itemize}

\end{rem}

\para

\begin{center}
\sf  Examples of Algorithm~\ref{alg1}
\end{center}

\para

\begin{exam}\label{exam1}
  Let $\S$ be  the algebraic surface defined by the polynomial $$F(x,y,z)=4x^2 + 9y^2 - 4x - 6y - z^2 + 2.$$

In Step~1, we compute $K(x,y,z)=576F(x,y,z)$, which means that $\S$ is a developable surface. In Step~2, we get that the tangent plane at $(u,v,w)$ is  $$(8u-4)(x-u)+(18v-6)y-v)-2w(z-w)=0, \quad \mbox{where}\quad F(u,v,w)=0.$$

One can find that some random tangent planes pass through $\P_0=(x_0,y_0,z_0)=(1/2,1/3,0)$. We check that $\P_0$ is  a common point of $ \{T(x,y,z; u,v,w)=0, F(u,v,w)=0\}$. Then,  $F(x,y,z)=0$ defines a conical surface.\\
Let $L(x,y,z)=x-z=0$ be the plane $\L$ not passing through $\P_0$. Then, the intersection curve of $\S$ and $\L$ is defined  by the polynomial $f(x,y)= 3x^2+9y^2-4x-6y+2$ (we eliminate the variable $z$). In addition, it has a rational proper parametrization given as
$$\left(\frac{9+t^2}{27+t^2},\frac{t^2-6t+27}{3t^2+81}\right)\in \mathbb{R}(t)^2.$$
Thus, a rational proper parametrization of the intersection space curve is
$$\left(\frac{9+t^2}{27+t^2},\frac{t^2-6t+27}{3t^2+81},\frac{9+t^2}{27+t^2}\right)\in \mathbb{R}(t)^3.$$
Finally, a parametrization of $\S$  is given by
$$(1-s)(1/2,1/3,0)+s\left(\frac{9+t^2}{27+t^2},\frac{t^2-6t+27}{3t^2+81},\frac{9+t^2}{27+t^2}\right) \in \mathbb{R}(s,t)^3.$$
\end{exam}

\para

\begin{exam}\label{exam2}
 Let $\S$ be the algebraic surface  defined by the polynomial\vspace*{2mm}

  \noindent
 $F(x,y,z)={x}^{4}+4\,{x}^{3}y+6\,{x}^{2}{y}^{2}+4\,x{y}^{3}+{y}^{4}-10\,{x}^{3}-
27\,{x}^{2}y-3\,{x}^{2}z-18\,x{y}^{2}-18\,xyz+6\,x{z}^{2}-2\,{y}^{3}-
12\,{y}^{2}z+3\,y{z}^{2}+{z}^{3}+16\,{x}^{2}+8\,xy+24\,xz+16\,{y}^{2}-
24\,yz+24\,{z}^{2}+64\,x-32\,y+96\,z$.

\vspace*{2mm}

We have that $K(x,y,z)=0$ on $\S$ which means that $\S$ is a developable surface.
There is not a fixed point that belongs to all the tangent planes. Then,  we go to Step 3.

\vspace*{2mm}

The vector $\N(x,y,z)$ of $\S$ is orthogonal to $\P_1=(1,-1,-1)$, i.e., $F_x-F_y-F_z=0$, which means that $F(x,y,z)=0$ defines a cylindrical surface.\\
 Let $L(x,y,z)=x+y+z=0$ be the plane not parallel to  $\P_1$. The intersection curve $\C$ of $\S$ and the plane  has a rational proper parametrization\\

\noindent
$\tilde\P(t)=(-554752\,{t}^{2}+439520\,t+311296\,{t}^{3}-65536\,{t}^{4}-130606,
65536\,{t}^{4}-307200\,{t}^{3}+540160\,{t}^{2}-422240\,t+123804,14592
\,{t}^{2}-17280\,t-4096\,{t}^{3}+6802
)\in \mathbb{R}(t)^3$.

\vspace*{2mm}

 \noindent
Then a parametrization of the surface $\S$ is given by $\tilde\P(t)+s(1,-1,-1)\in \mathbb{R}(s,t)^3$.
\end{exam}

\para

\begin{exam}\label{exam3}
Let  $\S$  be the  algebraic surface defined by the polynomial\vspace*{2 mm}

\noindent
$F(x,y,z)=11+16\,z-12\,y-36\,x-4\,{z}^{2}-48\,yz+12\,{y}^{2}-36\,xz+36\,xy+42\,{
x}^{2}+48\,{y}^{2}z+72\,xyz-24\,x{y}^{2}+24\,{x}^{2}z-36\,{x}^{2}y-20
\,{x}^{3}-32\,z{y}^{3}-48\,{y}^{2}zx-24\,zy{x}^{2}+12\,{x}^{2}{y}^{2}-
4\,z{x}^{3}+12\,{x}^{3}y+3\,{x}^{4}
$.

\vspace*{2mm}

We follow Steps 1,2 and 3, and we get that  $\S$ is not a conical surface or cylindrical surface but a tangential surface. Thus, we go to Step 4, and we find its cuspidal edge by solving $S=\{F_x=0,F_y=0,F_z=0,F=0\}$.
For this purpose, we use \texttt{WSOLVE} (\url{http://www.mmrc.iss.ac.cn/~dwang/wsolve.htm}) which is a maple package to solve the characteristic set. We get
$$S=\left\{\begin{array}{l}
  \{2\,xy+2\,yz-2\,y-3\,z+xz,4\,{y}^{3}-2\,z+6\,yz+{z}^{2}\};\\
 \{ 3+{x}^{2},y-1,z+2\};\\
\{2\,x-3,4\,y-1,4\,z-1\};\\
\{x-3,y,z-2\};\\
\{x-1,y,z\}.
\end{array}\right\}$$
Only the first component has dimension one, and then it should be the cuspidal edge. Thus, the cuspidal edge is an algebraic curve defined by two surfaces as $$\{2\,xy+2\,yz-2\,y-3\,z+xz=0;4\,{y}^{3}-2\,z+6\,yz+{z}^{2}=0\}.$$
We can determine its rationality and parameterized it by applying for instance the resolvent method in~\cite{Chou1992}. Actually, in this example, one can find that the cylindrical surface $4{y}^{3}-2\,z+6\,yz+{z}^{2}=0$ can be regarded as a planar curve (in ${\Bbb L}^2$), and it just has a rational parametrization~as $$
  (y,z)=\left(\frac{(t+2)(t-4)}{4(t-1)^2},\frac{(t+3)^3}{4(t-1)^3}\right).
$$
We substitute $(y,z)$ into the another surface $2\,xy+2\,yz-2\,y-3\,z+xz=0$, and we solve the variable $x$. Then, we get a parametrization of the  cuspidal edge given by
$$
  \tilde \P(t)=\left(\frac{3(t^2+2)}{2(t-1)^2},\frac{(t+2)(t-4)}{4(t-1)^2},\frac{(t+3)^3}{4(t-1)^3}\right)\in \mathbb{R}(t)^3.
$$
Finally, a rational parameterization of $\S$  is given by $\P(s,t)=  \tilde \P(t)+s  \tilde \P'(t)\in \mathbb{R}(s,t)^3$.
\end{exam}

\para

\begin{center}
  {\sf Refine the parameterizations}
\end{center}

\para

Reparametrizing a rational surface such that it does not contain any base point is usually a cumbersome
task, even for a ruled surface. An affine \emph{base point} of a rational surface parameterized by $\P(s,t)$ is a
parameter pair $(s_0,t_0)$ such that the numerator  and denominator  of each component of $\P(s,t)$ at $(s_0, t_0)$ are  zero. The $\mu$-basis technique in~\cite{chen03} provides   a simple and elegant way to reparameterize a rational ruled surface such that it does not contain any non-generic base point. Furthermore, the directrices of the reparameterized surface have the lowest possible degree. Thus there are both geometrical and
computational advantages to be gained from such a reparametrization. Here we refine the parametrization using the $\mu$-basis method in~\cite{chen03}. The more efficient algorithm to compute $\mu$-basis can be found in~\cite{Deng05}.
Continue to the Example~\ref{exam3} of the tangential surface, one can get the refined rational reparametrization $\P(u,v)=
\P_0(u)+v\P_1(u)\in \mathbb{R}(u,v)^3$ where $$\P_0(u)=\left(-{\frac {5\,u-8}{2({u}^{2}-2\,u+1)}},{\frac {{u}^{2}+7\,u-11}{4(
{u}^{2}-2\,u+1)}},-{\frac {7({u}^{2}+4\,u+4)}{8({u}^{2}-2\,u
+1)}}
\right)\in \mathbb{R}(u)^3$$ and $$\P_1(u)=(2\,{u}^{2}+2\,u-4,-3\,u+3,3\,{u}^{2}/2+6\,u+6)\in \mathbb{R}(u)^3.$$

\section{Parametrically developable surfaces }
In this section, we consider a surface $\S$ defined by a parametrization (not necessarily proper),
\begin{equation}\label{para}
  \P(s,t)=(p_1(s,t), p_2(s,t), p_3(s,t))\in {\Bbb L}(s,t)^3.
\end{equation}
We give the necessary and sufficient condition so that $\S$ represents a developable surface. The following theorem can be deduced from Theorem~\ref{developability}, and the details also was proposed in~\cite{burr50}.

\para

\begin{thm}\label{developability2}
  A given parametric surface $\S$ defined by $\P(s,t)$ of the form~\eqref{para},
  is developable if and only if $K(s,t)=0$, where
\begin{equation}\label{kst}
  K(s,t)=\left|\begin{array}{cccc}
l_{s} & l_{t} & l\\
m_{s} & m_{t} & m\\
n_{s} & n_{t} & n
\end{array}\right|,
\end{equation}
and $l=p_{2s}p_{3t}-p_{3s}p_{2t}$, $m=p_{3s}p_{1t}-p_{1s}p_{3t}$ and $n=p_{1s}p_{2t}-p_{2s}p_{1t}$.
\end{thm}

\para

Since we can tell whether $\S$, defined by $\P(s,t)$, is a developable surface,
we   compute a proper reparametrization in standard  form  for the affirmative case. For this purpose, let  $\N(s,t)=(n_1,n_2,n_3)=\P_s\times\P_t$ be the normal vector of the parametric surface $\S$, and we denote $\X=(x,y,z)$. Then,
the tangent plane of $\S$ is $T(x,y,z)=\N(s,t)\cdot(\X-\P(s,t))=0$.
According to the three types of developable surface (see Theorem~\ref{implicitdeve}), we have the following algorithm.

\para

\begin{alg}\label{alg2}{\sc Input}: An parametric surface $\S$ defined by the parametrization $\P(s,t)$.\\
{\sc Output}: A proper rational parametrization $\P(s,t)$ of  $\S$ or the message of ``$\S$ is not a developable surface".
 \begin{itemize}
   \item[1.] Compute $K(s,t)$ of the form~\eqref{kst}. If it is zero go to Step 2. Otherwise, {\sf Return} ``$\S$ is not a developable surface."
   \item[2.] If $T(x,y,z)=0$ passes through a fixed point $\P_0$,  let $\tilde\P(t)$ be a rational proper parametrization of the curve $\C$ which is the intersection curve of $\S$ and a plane not passing through $\P_0$.  Then, {\sf Return} $ (1-s)\P_0+ s \tilde \P(t)$ is a rational proper parametrization  of  the conical surface $\S$.

 \item[3.] If there exists  $(0,0,0)\neq\P_1=(p_{11},p_{12},p_{13})\in \mathbb{R}^3 $ such that $p_{11}n_1+ p_{12}n_2+p_{13}n_3=0$,  let $\tilde\P(t)$ be a rational proper parametrization of the curve  $\C$ which is the intersection of $\S$
           and a plane not parallel to $\P_1$. Then, {\sf Return} $\tilde \P(t) + s\P_1$ is a rational proper parametrization  of  the cylindrical surface $\S$.

\item[4.] Solve the algebraic system $S=\{\X-\P(s,t)=(0,0,0),\N(s,t)=(0,0,0)\}$, and find a rational proper parametrization $\tilde\P(t)$ of a curve $\C\in S$ (apply for instance the resolvent method in~\cite{Chou1992}). If $\tilde\P(t)+s\tilde\P'(t)$ parametrizes $\S$,   {\sf Return} $\tilde\P(t)+s\tilde\P'(t)$ is a rational proper parametrization of~$\S$.
 \end{itemize}
\end{alg}

\para

\begin{rem} We give some necessary remarks for Algorithm~\ref{alg2}.
  \begin{itemize}\label{rem2}
    \item[a.] The point $\P_0$ in Step~2  can be obtained from the coefficient set of the tangent plane $T(x,y,z; s,t)=0$ with respect to $\{s,t\}$. To compute $\P_1$ in Step~3, we just need to find the linearly dependent coefficient vector of  $n_1, n_2$ and $n_3$.
    \item[b.]   In Step~4, it is known that the cuspidal edge is included in the singular set $S$ of the surface $\S$. The cuspidal edge is a prime set and then, it can be separately  solved by the resolvent method in~\cite{Chou1992}.
    \item[c.] For the intersection curve in Step 2 or Step 3,  or the singular curve  in Step 4, we may get an improper parameterized curve if the given parametrization $\P(s,t)$ is improper. In this case, we  check out the improper case, and we properly reparameterize the curve by some methods such as~\cite{Perez-repara, sendra04, Sed86}.
    \item[d.] In Step~4, for a rational parametrization $\tilde\P(t)+s\tilde\P'(t)$, we should check whether it is a reparametrization of the given parametrization $\P(s,t)$. Here, we recommend to implicitize $\tilde\P(t)+s\tilde\P'(t)$. Observe that $\S$ is a ruled surface and then it has an efficient implicitization method (see for instance~\cite{shen10}). One can check the result by substituting the given parametrization into the implicit equation.
       \end{itemize}
  \end{rem}

\para

\para

\begin{center}
\sf  Examples of Algorithm~\ref{alg2}
\end{center}

\para

\begin{exam}Let $\S$ be a parametric surface defined by
$$\P(s,t)=\left({\frac {4\,{s}^{2}+t+1-2\,s+{t}^{2}+2\,ts}{1-2\,t-2\,s+{t}^{2}+2\,ts+
{s}^{2}}},{\frac {6\,t{s}^{2}+7\,{t}^{2}+6\,{s}^{3}+8\,ts-{s}^{2}-4\,t
+1-2\,s}{1-2\,t-2\,s+{t}^{2}+2\,ts+{s}^{2}}},\right.$$
$$\left. {\frac {{t}^{2}{s}^{2}+2
\,t{s}^{3}+6\,t{s}^{2}+{t}^{3}+2\,{t}^{2}s+5\,{t}^{2}+{s}^{4}+5\,{s}^{
3}+5\,ts}{1-2\,t-2\,s+{t}^{2}+2\,ts+{s}^{2}}}\right)\in \mathbb{R}(s,t)^3.$$

Following Step~1, $\S$ is a developable surface since $K(s,t)=0$. We can find this parametrization is improper (see the method in~\cite{sendra04}).

\vspace*{2mm}

In Step~2, solving the coefficient set of $T(x,y,z; s,t)$ with respect to $\{s,t\}$, we get a fixed point $\P_0=(1,1,0)$ of the tangent planes. Therefore, $\P(s,t)$ is a conical surface with the apex $\P_0$.

\vspace*{2mm}

 Let $L(x,y,z)=0$ be the equation defining a plane $\L$ not passing through $\P_0$. To simplify the computation,  one can find a plane $L(x,y,z)\in {\Bbb L}[var]$, where  $var\in \{x, y,z\}$. Then, $\{\X=\P(s,t), L(\P(s,t))=0\}$ defines a planar curve $\C\subset {\Bbb L}^3$.  In this example, we set $L(x,y,z)=z-1=0$ and then, the curve $\C$ is included in  $\{x=p_1(s,t),y=p_2(s,t),z=p_3(s,t),z-1=0\}$.  By the classic properties of the resultant (see for instance \cite{Cox1998} or \cite{Vander}), the implicit equation of the projected curve of $\C$ on the $(x:y)$ plane is a factor of $$\Res_s(\Res_t(\num(x-p_1),\num(p_3-1)),\Res_t(\num(y-p_2),\num(p_3-1))),$$ where $\num(\cdot )$ returns the numerator of a rational function, and $\Res_{var}$ returns the resultant of two polynomials with respect to $var$. We get that the implicit equation of the projected curve is $$283-338\,x+64\,{x}^{2}-120\,y+102\,xy+9\,{y}^{2}=0.$$ We find a rational parametrization, and we lift it to get the parametrization of $\C$. We have that
$$\tilde\P(t)=\left({\frac {-283+120\,t-9\,{t}^{2}}{-258+90\,t}},{\frac {283-507\,t+144\,
{t}^{2}}{-387+135\,t}},1\right)\in \mathbb{R}(t)^3.
$$

Finally, we obtain a rational proper reparametrization $ (1-s)\P_0+ s \tilde \P(t)\in \mathbb{R}(s,t)^3$ for the surface $\S$.\\

\noindent
The computation process is similar as above if the surface has a ruling direction.
\end{exam}

\para

\begin{exam}
Let $\S$ be a parametric surface defined by the parametrization\vspace*{2mm}

 \noindent
$\P(s,t)=((-1+2\,t+2\,s+3\,{t}^{2}{s}^{2}-2\,ts-t{s}^{2}+2\,t{s}^{3}+4\,{s}^{5}-{
t}^{6}+4\,{s}^{4}{t}^{2}-3\,{t}^{4}{s}^{2}-2\,{t}^{2}{s}^{3}-2\,{t}^{4
}s+4\,{s}^{4}t-2\,{t}^{3}{s}^{2}-2\,{t}^{3}s-{s}^{3}-{s}^{4}-2\,{t}^{5
}-{s}^{2}
)/ \left( {t}^{2}+s+t-1 \right) ^{2},
(-3\,{t}^{4}s-2\,{t}^{2}{s}^{2}+3\,{t}^{2}s+4\,t{s}^{3}-5\,{t}^{5}-t{s}
^{2}-{t}^{2}+3\,{t}^{3}+2\,{s}^{4}-{s}^{3}-3\,{s}^{4}t-6\,{t}^{3}s+2\,
{t}^{4}{s}^{2}-6\,{t}^{3}{s}^{3}+6\,{t}^{3}{s}^{2}+2\,{t}^{2}{s}^{3}+6
\,{t}^{5}s-{s}^{5}-2\,{s}^{6}-3\,{t}^{4}{s}^{4}+3\,{t}^{2}{s}^{6}+3\,{
s}^{7}+3\,{t}^{6}s-3\,{t}^{5}{s}^{2}-{t}^{6}{s}^{2}+3\,{s}^{6}t-3\,{s}
^{4}{t}^{3}+{t}^{8}-3\,{t}^{4}{s}^{3}+3\,{t}^{7}-3\,{t}^{2}{s}^{5}
)/ \left( {t}^{2}+s+t-1 \right) ^{3},
2\,{s}^{4} ( 3\,{t}^{2}{s}^{2}+3\,{s}^{3}+3\,t{s}^{2}-2\,{s}^{2}-
3\,{t}^{4}-3\,{t}^{2}s-3\,{t}^{3}+3\,{t}^{2} )
/\left( {t}^{2}+s+t-1 \right) ^{3})\in \mathbb{R}(s,t)^3$.\vspace*{2mm}

$\S$ is a developable surface because $K(s,t)=0$.  In addition, $\S$ is a rational tangential surface since it is not a conical or cylindrical surface.

\vspace*{2mm}

In Step~4, we solve the algebraic system $S=\{\X-\P(s,t)=(0,0,0),\N(s,t)=(0,0,0)\}$ using \textsf{Maple} package \texttt{WSOLVE}. We get a rational  space curve defined by the proper parametrization
$$\tilde\P_k(t)=(2\,{t}^{2}-3\,t,-{t}^{3}+3/2\,{t}^{2}+1/4\,t-3/8,-2\,{t}^{3}+3\,{t}^{
2}-3/2\,t+1/4
)\in \mathbb{R}(t)^3.$$

A proper rational tangential surface $\S_k$ can be constructed as $\P_k(s,t)=\tilde\P_k(t)+s\tilde\P_k'(t)\in \mathbb{R}(s,t)^3$. Using univariate resultant (see~\cite{shen10}),
we can get its implicit equation\vspace*{2mm}

 \noindent
$32+96\,y+96\,x-48\,z+96\,{x}^{2}+32\,{x}^{3}+48\,{y}^{2}-64\,{y}^{3}-
48\,{y}^{4}-48\,{y}^{2}{x}^{2}+96\,{x}^{2}y+192\,xy-96\,x{y}^{3}-20\,{
z}^{2}+32\,{z}^{3}+13\,{z}^{4}+48\,zy+12\,{z}^{3}x-12\,{z}^{2}{x}^{2}-
48\,{z}^{2}x-144\,{z}^{2}y+120\,{z}^{2}{y}^{2}-72\,{z}^{3}y-32\,z{y}^{
3}+192\,z{y}^{2}-48\,z{x}^{2}-72\,{z}^{2}xy+144\,{y}^{2}xz+96\,zyx+48
\,z{x}^{2}y-96\,zx
=0.$
\vspace*{2mm}

It holds that this surface cover the parametric surface $\S$ defined by $\P(s,t)$. Hence the $\P_k(s,t)$ is a rational proper reparametrization of the given one.
\end{exam}

\section{Conclusion}
The developability of an algebraic surface associates with the Gaussian curvature. There are three different types of developable surface, we discuss each type of them and then the problem is simplified. For a developable surface, we determine its rationality by discussing the three types of developable surfaces. We prove of the main theorem constructively. A rational proper parametrization of the rational developable surfaces is the proposed. For a rational parametrization (not necessarily proper), we determine its developability and find a proper reparametrization for the developable one.

\end{document}